\documentclass[11pt]{article}  

\usepackage[utf8]{inputenc} 
\usepackage[T1]{fontenc}    
\usepackage{amssymb}
\usepackage{amsmath}
\usepackage{amsthm}
\usepackage{color}
\usepackage{colortbl}
\usepackage{graphicx}
\usepackage{hyperref}
\usepackage{fullpage}
\usepackage[linesnumbered, boxed]{algorithm2e}

\newcommand{\R}{\mathbb{R}}

\newcommand{\Z}{\mathbb{Z}}
\newcommand{\E}{\mathbb{E}}

\newcommand{\bx}{{\bf x}}
\newcommand{\by}{{\bf y}}
\newcommand{\bz}{{\bf z}}
\newcommand{\bv}{{\bf v}}
\newcommand{\bw}{{\bf w}}

\newtheorem{theorem}{Theorem}
\newtheorem{lemma}[theorem]{Lemma}
\newtheorem{corollary}[theorem]{Corollary}

\newcommand{\eps}{\epsilon}

\begin{document}
\title{A Reduction for Optimizing Lattice Submodular Functions with Diminishing Returns}
\author{
Alina Ene \thanks{Department of Computer Science and DIMAP,
University of Warwick, \texttt{A.Ene@warwick.ac.uk}.}
\and
Huy L. Nguyen\thanks{Toyota Technological Institute at Chicago,
{\tt hlnguyen@cs.princeton.edu}.}
}

\maketitle

\begin{abstract}
A function $f: \mathbb{Z}_+^E \rightarrow \mathbb{R}_+$ is DR-submodular if it satisfies $f(\bx + \chi_i) -f (\bx) \ge f(\by + \chi_i) - f(\by)$ for all $\bx\le \by, i\in E$. Recently, the problem of maximizing a DR-submodular function $f: \mathbb{Z}_+^E \rightarrow \mathbb{R}_+$ subject to a budget constraint $\|\bx\|_1 \leq B$ as well as additional constraints has received significant attention \cite{SKIK14,SY15,MYK15,SY16}.

In this note, we give a generic reduction from the DR-submodular setting to the submodular setting. The running time of the reduction and the size of the resulting submodular instance depends only \emph{logarithmically} on $B$. Using this reduction, one can translate the results for unconstrained and constrained submodular maximization to the DR-submodular setting for many types of constraints in a unified manner.
\end{abstract}

\section{Introduction}
Recently, constrained submodular optimization has attracted a lot of attention as a common abstraction of a variety of tasks in machine learning ranging from feature selection, exemplar clustering to sensor placement. Motivated by the use cases where there is a large budget of identical items, a generalization of submodular optimization to integer lattice is proposed by~\cite{SKIK14}. Previously, submodular functions has been generalized to lattices via the lattice submodular property. A function $f:\Z^E_+ \to \R_+$ is {\em lattice submodular} if for all $\bx,\by\in \Z^E_+$,
$$f(\bx)+f(\by) \ge f(\bx \vee \by) + f(\bx \wedge \by)$$

In the generalization due to~\cite{SKIK14,SY15}, a function $f$ is {\em DR-submodular} if it satisfies
$$f(\bx + \chi_i) -f (\bx) \ge f(\by + \chi_i) - f(\by)$$
for all $\bx\le \by, i\in E$ (diminishing return property), where $\chi_i$ is the vector in $\{0, 1\}^E$ that has a $1$ in the coordinate corresponding to $i$ and $0$ in all other coordinates.

It can be shown that any DR-submodular function is also lattice submodular (but the reverse direction is not necessarily true). Similar to submodular functions, the applications can be formulated as maximizing a DR-submodular function $f$ subject to constraints, such as a budget constraint $\max \{f(\bx) \colon \bx \in \mathbb{Z}_+^{E}, \|\bx\|_1 \le B\}$. While it is straightforward to reduce optimization of DR-submodular function with budget constraint $B$ to optimization of submodular function with $B\cdot E$ items, the goal of \cite{SKIK14} is to find algorithms for this setting with running time \emph{logarithmic} in $B$ rather than \emph{polynomial} in $B$, which follows from the straightforward reduction. Following \cite{SKIK14}, there have been several works extending problems involving submodular functions to the DR-submodular setting~\cite{SY15,MYK15,SY16}.

In this note, we give a generic reduction from the DR-submodular setting to the submodular setting. The running time of the reduction and the size of the resulting submodular instance depends only \emph{logarithmically} on $B$. Using this reduction, one can translate the results for unconstrained and constrained submodular maximization to the DR-submodular setting for many types of constraints in a unified manner.

\section{The Reduction}

\begin{lemma}
\label{lem:decomposition}
For any $n$, there is a decomposition $n = a_1 + a_2 +\ldots + a_t$ with $t\le 2\log n+1$ so that for any $q \le n$, there is a way to express $q$ as the sum of a subset of the multiset $\{a_1, \ldots, a_t\}$.
\end{lemma}

\begin{proof}
Let $n = b_0 2^0 + b_1 2^1 + \ldots + b_m 2^m$ be the binary representation of $n$ with $b_m=1$. Let $b_{c_1}, b_{c_2}, ..., b_{c_p}$ be all the non-zeroes among $b_0, \ldots, b_{m-1}$. Let $a_1=1$, $a_i = 2^{i-2}$ for $2\le i \le m+1$, and $a_{m+1+j} =b_{c_j} 2^{c_j}$ for $1\le j \le p$. It is clear that $\sum_{i\le m+1} a_i = 2^m$ and $\sum_{i} a_i = n$.

Consider an arbitrary number $1 < q < n$.	Let $j$ be the largest bit that is 1 for $n$ but it is 0 for $q$ ($j$ must exist because $q < n$). Let $r$ be the number that agrees with $n$ on all bits larger or equal to $j$ and has all 0s for the smaller bits. We can form $r$ from $a_1,\ldots, a_{m+1}$ (which sum up to $2^m$) and the additional numbers from $\{a_{m+2}, \ldots, a_{m+1+p}\}$ corresponding to the bits equal to 1 from $j$ to $m-1$ in the binary representation of $r$. Notice that $r-q < 2^m$ and it can be written as a sum of numbers from $a_2, \ldots, a_{m+1}$ (just $(r-q)$'s binary representation). By removing those numbers from the representation of $q$ above, we obtain a subset of the $a_i$'s that sums to $q$.
\end{proof}

\begin{corollary}
\label{cor:decomposition}
For any $n$, there is a way to write $n = a_1 + a_2 +\ldots + a_t$ with $t\le 2\log n+1+1/\eps$ so that $a_i \le \eps n~\forall i$ and for any $q \le n$, there is a way to express $q$ as the sum of a subset of the multiset $\{a_1, \ldots, a_t\}$.
\end{corollary}
\begin{proof}
We start with the decomposition of the above lemma and refine it until the condition $a_i\le \eps n~\forall i$ is satisfied. As long as there exists some $a_i > \eps n$, replace $a_i$ with two new numbers $a_i-\eps n$ and $\eps n$. Each replacement step produces a new term equal to $\eps n$ so the number of replacement steps is at most $1/\eps$. Thus, the number of terms in the decomposition is at most $2\log n + 1 + 1/\eps$.
\end{proof}

\paragraph{The reduction.}
Suppose we need to optimize $f$ over the domain $[B_1]\times[B_2]\times\cdots\times[B_E]$. By the above lemma, we can write $B_i = a_{i,1}+\ldots + a_{i, t_i}$ with $t_i \le 2\log B_i + 1$ and any number at most $B_i$ can be written as a sum of a subset of the $\{a_{i,j}\}_j$'s. Let $t = \sum_i t_i$. Consider a function $g$ defined on the ground set $E' = \bigcup_{i\in E} \{(i, 1), \ldots, (i, t_i)\}$ defined as follows. Consider $\by \in \{0,1\}^{E'}$. Let $x_i = \sum_{j} y_{i, j} a_j$ and we define $g(\by) := f(\bx)$.

By Lemma~\ref{lem:decomposition}, for any vector $\bx$, there is a vector $\by$ such that $x_i = \sum_{j} y_{i, j} a_{i,j}$ for all $i$. Thus, the set $\{g(\by)\}_{\by}$ captures all of $\{f(\bx)\}_{\bx}$. Next, we show that $g$ is submodular.

\begin{lemma}
The function $g$ is submodular.
\end{lemma}

\begin{proof}
Consider 2 vectors $\by, \by' \in \{0,1\}^{E'}$ such that $y_{i,j} \le y'_{i,j}$ for all $i, j$. Consider an arbitrary element $(i_0,j_0) \in E'$ that is not in $\by'$. Let $\bx$ defined as $x_i = \sum_{j} y_{i, j} a_{i,j}$ and $\bx'$ defined as $x'_i = \sum_{j} y'_{i, j} a_{i,j}$. We have $g(\by + \chi_{(i_0, j_0)}) - g(\by) = f(\bx + a_{i_0,j_0}\chi_{i_0}) - f(\bx)$ and $g(\by' + \chi_{(i_0, j_0)}) - g(\by') = f(\bx' + a_{i_0,j_0}\chi_{i_0}) - f(\bx')$. By the diminishing return property, we have
$$f(\bx + a_{i_0,j_0}\chi_{i_0}) - f(\bx) \ge f(\bx' + a_{i_0,j_0}\chi_{i_0}) - f(\bx').$$
\end{proof}

\section{Modeling Constraints}

We are interested in maximizing $f(\bx)$ subject to constraints. In this section, we show how to translate constraints on $\bx$ to constraints for maximizing $g(\by)$.

\paragraph{Cardinality constraint.} The constraint $\sum_i x_i \le K$ with $K > 1/\eps$ can be translated to
$$\sum_{i,j} a_{i,j} y_{i,j} \le K.$$
By applying Corollary~\ref{cor:decomposition}, we map from a cardinality constraint to a knapsack constraint where all weights are at most an $\eps$ fraction of the budget.

\paragraph{Knapsack constraint.} The knapsack constraint $\sum_i c_i x_i \le K$ can be translated to
$$\sum_{i,j} c_i a_{i,j} y_{i,j} \le K.$$

\paragraph{General constraints.} Consider the problem $\max\{f(\bx) \colon \bx \in \mathcal{I}\}$, where $\mathcal{I} \subseteq [B_1] \times [B_2] \times \ldots \times [B_E]$ denotes the set of all solutions that satisfy the constraints.

We can apply algorithmic frameworks from the submodular setting --- such as the frameworks based on continuous relaxations and rounding \cite{Vondrak08,ChekuriVZ14} --- to the DR-submodular setting as follows. Let $\mathcal{P} \subseteq \mathbb{R}_+^{E}$ be a relaxation of $\mathcal{I}$ that satisfies the following conditions:
\begin{itemize}
  \item $\mathcal{P}$ is downward-closed: if $\bx \leq \bz$ and $\bz \in \mathcal{P}$ then $\bx \in \mathcal{P}$.
  \item There is a separation oracle for $\mathcal{P}$: given $\bx$, there is an oracle that either correctly decides that $\bx \in \mathcal{P}$ or otherwise returns a hyperplane separating $\bx$ from $\mathcal{P}$, i.e., a vector $\bv \in \mathbb{R}^{E}$ and $D \in \mathbb{R}$ such that $\left<\bv, \bx\right> \geq D$ and $\left<\bv, \bz \right> <  D$ for all $\bz \in \mathcal{P}$.
\end{itemize}

We apply Lemma~\ref{lem:decomposition} (or Corollary~\ref{cor:decomposition}) to obtain the multiset $\{a_{i, j}\}_{i, j}$ such that, for any vector $\bx$, there is a vector $\by$ such that $x_i = \sum_{j} y_{i, j} a_{i,j}$ for all $i$. Define the linear function $M:\R^{E'} \to \R^{E}$ where $\bx = M(\by)$ is computed according to $x_i = \sum_{j} y_{i, j} a_{i,j}~\forall i$. Let $g: 2^{E'} \to \mathbb{R}_+$ be the submodular function given by the reduction. Let $G: [0, 1]^{E'} \rightarrow \mathbb{R}_+$ be the multilinear extension of $g$:
  $$G(\by) = \E[g(R(\by))],$$
where $R(\by)$ is a random set that contains each element $e \in E'$ independently at random with probability $y_e$.

Thus we obtain the following fractional problem: $\max\{G(\by) \colon \by\in [0,1]^{E'}, M(\by) \in \mathcal{P}\}$. As shown in the following lemma, we can use the separation oracle for $\mathcal{P}$ to maximize a linear objective $\left<\bw, \by\right>$, where $\bw \in \mathbb{R}^{E'}$, subject to the constraints $\by\in[0,1]^{E'}$ and $M(\by) \in \mathcal{P}$.

\begin{lemma} \label{lem:sep-oracle}
	Using the separation oracle for $\mathcal{P}$ and an algorithm such as the ellipsoid method, for any vector $\bw \in \mathbb{R}^{E'}$, one can find in polynomial time a vector $\by \in \mathbb{R}^{E'}$ that maximizes $\left<\bw, \by\right>$ subject to $\by\in [0,1]^{E'}$ and $M(\by) \in \mathcal{P}$.
\end{lemma}
\begin{proof}
	It suffices to verify that the separation oracle for $\mathcal{P}$ allows us to separate over $\{\by \colon \by\in[0,1]^{E'}, M(\by) \in \mathcal{P}\}$. To this end, let $\by$ be a vector in $\mathbb{R}^{E'}$. Separation for the constraint $\by\in[0,1]^{E'}$ can be done trivially by checking if every coordinate of $\by$ is in $[0,1]$. Thus, we focus on separation for the constraint $M(\by)\in \mathcal{P}$. Using the separation oracle for $\mathcal{P}$, we can check whether $M(\by) \in \mathcal{P}$. If yes, then we are done. Otherwise, the oracle returns $\bv \in \mathbb{R}^E$ and $D \in \mathbb{R}$ such that $\left<\bv, M(\by)\right> \geq D$ and $\left<\bv, \bz\right> < D$ for all $\bz \in \mathcal{P}$. Let $\bv' \in \mathbb{R}^{E'}$ be $\bv' = M^*(\bv)$, where $M^*$ is the adjoint of $M$ i.e. $v'_{i,j} = a_{i,j} v_i~\forall i,j$. Then $\left<\bv', \by\right> = \left<M^*(\bv), \by\right> = \left<\bv, M(\by)\right> \geq D$. Now let $\by'$ be a vector in $\mathbb{R}^{E'}$ such that $M(\by') \in \mathcal{P}$. We have $\left<\bv',\by'\right> = \left<M^*(\bv), \by'\right> = \left<\bv, M(\by')\right> < D$. Thus $(\bv', D)$ is a hyperplane separating $\by$ from $\{\by' \colon M(\by') \in \mathcal{P}\}$.
\end{proof}

Since we can solve $\max\{\left<\bw, \by\right> \colon \by\in[0,1]^{E'}, M(\by) \in \mathcal{P}\}$, where $\bw \in \mathbb{R}^{E'}$, we can approximately solve the fractional problem $\max\{G(\by) \colon \by\in[0,1]^{E'}, M(\by) \in \mathcal{P}\}$ using the (measured) Continuous Greedy algorithm or local search \cite{Vondrak08,ChekuriVZ14,FeldmanNS11}.

We note that in some settings, such as when $\mathcal{P}$ is a polymatroid polytope\footnote{Let $\rho: 2^E \rightarrow \mathbb{Z}_+$ be a monotone submodular function with $\rho(\emptyset) = 0$. The polymatroid associated with $\rho$ is the polytope $\mathcal{P} = \{\bx \in \mathbb{R}^E_+ \colon \sum_{i \in S} x_i \leq \rho(S) \quad \forall S \subseteq E\}$.}, we can round the resulting fractional solution to the problem $\max\{G(\by) \colon M(\by) \in \mathcal{P}\}$ and obtain an integral solution (similarly to \cite{SY16}); in this case, the rounding preserves the value of the fractional solution and thus we obtain an $\alpha$-approximation for the problem $\max\{f(\bx) \colon \bx \in \mathcal{I}\}$, where $\alpha = 1 - 1/e$ for monotone functions and $\alpha = 1/e$ for non-monotone functions. The detailed proof is in Theorem~\ref{thm:polymatroid}. 

\paragraph{Some examples of results.} Using the reduction above, we immediately get algorithms for maximizing DR-submodular functions subject to various types of constraints. We include a few examples below.

\begin{theorem}
There is a $1/2$ approximation algorithm for unconstrained DR-submodular maximization with running time $O(n+\sum_i \log B_i)$.
\end{theorem}
\begin{proof}
By the reduction using Lemma~\ref{lem:decomposition}, we need to solve an unconstrained submodular maximization with $O(n+\sum_i \log B_i)$ items. The result follows from applying the Double Greedy algorithm of \cite{BuchbinderFNS15} to the resulting instance of unconstrained submodular maximization.
\end{proof}

\begin{theorem}
There is a $1-1/e-\eps$ approximation algorithm for maximizing a monotone DR-submodular function subject to a cardinality constraint $B$ with running time $O(m\log(m/\eps)/\eps)$ where $m = n/\eps +\sum_i \log B_i$.
\end{theorem}
\begin{proof}
If $B \le 1/\eps$, the result follows via the trivial reduction of making $B$ copies of every item. Next, we consider the case $B>1/\eps$.
By the reduction using Corollary~\ref{cor:decomposition}, we need to solve a submodular maximization problem with a knapsack constraint where all weights are at most $\eps$ times the budget and there are $O(n/\eps+\sum_i \log B_i)$ items. The result follows from applying the Density Greedy algorithm with either descending thresholds or lazy evaluation \cite{BadanidiyuruV14}.
\end{proof}

\begin{theorem}
	There is an $\alpha$ approximation algorithm for maximizing a DR-submodular function subject to a polymatroid constraint with running time that is polynomial in $n$ and $\sum_i \log B_i$, where $\alpha = 1 - 1/e$ if the function is monotone and $\alpha = 1/e$ otherwise.
	\label{thm:polymatroid}
\end{theorem}
\begin{proof}
	Let $\mathcal{P}$ be the polymatroid polytope. We apply Lemma~\ref{lem:decomposition} (or Corollary~\ref{cor:decomposition}) to obtain the multiset $\{a_{i, j}\}_{i, j}$ such that, for any vector $\bx$, there is a vector $\by$ such that $x_i = \sum_{j} y_{i, j} a_j$ for all $i$. Let $g: 2^{E'} \rightarrow \mathbb{R}_+$ be the submodular function given by the reduction. Let $G: [0, 1]^{E'} \rightarrow \mathbb{R}_+$ be the multilinear extension of $g$.
	
	Since we can separate over $\mathcal{P}$ in polynomial time using a submodular minimization algorithm, it follows from Lemma~\ref{lem:sep-oracle} that we can optimize any linear (in $\by$) objective over $\bx \in \mathcal{P}$ and $x_i \leq B_i$ for all $i \in E$, where $\bx=M(\by)$. Therefore, using the measured Continous Greedy algorithm, we can find an $\alpha$-approximate fractional solution to the problem $\max\{G(\by) \colon \bx \in \mathcal{P}, x_i \leq B_i \; \forall i \in E\}$, where $\alpha = 1 - 1/e$ for monotone functions and $\alpha = 1/e$ for non-monotone functions. Similarly to \cite{SY16}, we can round the resulting fractional solution without any loss in the approximation. Let $\bz \in \Z^{E}$ be defined as $z_i = \lfloor x_i \rfloor$. Define $H:[0,1]^E \to \R$ as $H(\bv) = \E_{R}[f(\bz+R(\bv))]$ for any $\bv\in \R^E$. Note that $H$ is the multilinear extension of a submodular function agreeing with $H$ on $\{0,1\}^E$.

Let $\bv\in \R^E$ be defined as $v_i = x_i - \lfloor x_i\rfloor$.
First one can show that $H(\bv) \ge G(\by)$ via a hybrid argument. Let $\bz^{(i)} \in \Z^{E}$ be a random integral vector whose first $i$ coordinates are distributed according to $M(R(\by))$ (that is, constructing a randomized rounding of $\by$ and then converting it to an integral vector in $\Z^E$) and the last $|E|-i$ coordinates are picked randomly among $\{z_i, z_i+1\}$ so that the expectation is $x_i$. Note that $\E[f(\bz^{(0)})] = H(\bv)$ and $\E[f(\bz^{(|E|)})] = G(\by)$ and we will show that $\E[f(\bz^{(i-1)})] \ge \E[f(\bz^{(i)})]$. Indeed, for all $j\ne i$, $\bz^{(i)}_j$ and $\bz^{(i-1)}_j$ are identically distributed so we can couple the randomness so that $\bz^{(i)}_j = \bz^{(i-1)}_j~\forall j\ne i$. Let $\bw$ be $\bz^{(i)}$ with the $i$th coordinate zeroed out and define a single variable function $g:\Z\to \R$ where $g(x) = f(\bw+x\cdot e_i)$. Define $g':\R\to \R$ be the piecewise linear function agreeing with $g$ on integral points and $g'$ does not have any break points other than integral points. By the DR property of $f$, we have that $g'$ is concave. Thus, $\E[g'(\bz^{(i)}_i)] \le g'(\E[\bz^{(i)}_i])$. On the other hand, because $g'$ is linear in $[z_i, z_i+1]$, we have $\E[g'(\bz^{(i-1)}_i)] = g'(\E[\bz^{(i-1)}_i]) = g'(\E[\bz^{(i)}_i])$.

Next as done in \cite[Lemma 13]{SY16}, one can show that the constraints $\bv\in [0,1]^E, \bz+\bv \in \mathcal{P}$ are equivalent to a matroid polytope. Thus, one can round $\bv$ to an integral vector without losing any value in $H$ using strategies such as pipage rounding or swap rounding.
\end{proof}
\bibliography{dr}

\begin{thebibliography}{1}

\bibitem{BadanidiyuruV14}
Ashwinkumar Badanidiyuru and Jan Vondr{\'{a}}k.
\newblock Fast algorithms for maximizing submodular functions.
\newblock In {\em Proceedings of the Twenty-Fifth Annual {ACM-SIAM} Symposium
  on Discrete Algorithms (SODA)}, pages 1497--1514, 2014.

\bibitem{BuchbinderFNS15}
Niv Buchbinder, Moran Feldman, Joseph Naor, and Roy Schwartz.
\newblock A tight linear time (1/2)-approximation for unconstrained submodular
  maximization.
\newblock {\em {SIAM} J. Comput.}, 44(5):1384--1402, 2015.

\bibitem{ChekuriVZ14}
Chandra Chekuri, Jan Vondr{\'{a}}k, and Rico Zenklusen.
\newblock Submodular function maximization via the multilinear relaxation and
  contention resolution schemes.
\newblock {\em {SIAM} J. Comput.}, 43(6):1831--1879, 2014.

\bibitem{FeldmanNS11}
Moran Feldman, Joseph Naor, and Roy Schwartz.
\newblock A unified continuous greedy algorithm for submodular maximization.
\newblock In {\em Proceedings of the 52nd Annual {IEEE} Symposium on
  Foundations of Computer Science (FOCS)}, pages 570--579, 2011.

\bibitem{MYK15}
Takanori Maehara, Akihiro Yabe, JP~NEC, and Ken-ichi Kawarabayashi.
\newblock Budget allocation problem with multiple advertisers: A game theoretic
  view.
\newblock In {\em Proceedings of the 32nd International Conference on Machine
  Learning (ICML)}, pages 428--437, 2015.

\bibitem{SKIK14}
Tasuku Soma, Naonori Kakimura, Kazuhiro Inaba, and Ken-ichi Kawarabayashi.
\newblock Optimal budget allocation: Theoretical guarantee and efficient
  algorithm.
\newblock In {\em Proceedings of The 31st International Conference on Machine
  Learning (ICML)}, pages 351--359, 2014.

\bibitem{SY15}
Tasuku Soma and Yuichi Yoshida.
\newblock A generalization of submodular cover via the diminishing return
  property on the integer lattice.
\newblock In {\em Advances in Neural Information Processing Systems (NIPS)},
  pages 847--855, 2015.

\bibitem{SY16}
Tasuku Soma and Yuichi Yoshida.
\newblock Maximizing monotone submodular functions over the integer lattice.
\newblock In {\em International Conference on Integer Programming and
  Combinatorial Optimization (IPCO)}, pages 325--336, 2016.

\bibitem{Vondrak08}
Jan Vondr{\'{a}}k.
\newblock Optimal approximation for the submodular welfare problem in the value
  oracle model.
\newblock In {\em Proceedings of the 40th Annual {ACM} Symposium on Theory of
  Computing (STOC)}, pages 67--74, 2008.

\end{thebibliography}
\bibliographystyle{plain}
\end{document}